\documentclass{llncs}

\usepackage{times}
\usepackage{latexsym}
\usepackage{amssymb,amsmath}
\usepackage{enumitem}
\usepackage{nameref}
\usepackage[group-separator={,}]{siunitx}
\usepackage{complexity}

\newtheorem{mytheorem}{Theorem}
\newtheorem{mylemma}{Lemma}
\newtheorem{myreduction}{Reduction}
\newtheorem{myobservation}{Observation}
\newcommand{\myqed}{\mbox{$\Box$}}

\newcommand{\mylike}{\textsc{Li\-ke}}
\newcommand{\myblike}{\textsc{Ba\-la\-n\-ced Li\-ke}}

\newcommand{\mynecessary}{\mbox{\sc N\-e\-ce\-s\-sa\-ry}}
\newcommand{\myexact}{\mbox{\sc Ex\-act}}
\newcommand{\myutility}{\mbox{\sc U\-ti\-li\-ty}}
\newcommand{\mystochastic}{\mbox{\sc Sto\-ch\-as\-tic}}

\begin{document} 

\title{Expected Outcomes and Manipulations \\ in Online Fair Division}
\author{Martin Aleksandrov \and Toby Walsh}
\institute{UNSW Sydney and Data61, CSIRO and TU Berlin \email{\{martin.aleksandrov,toby.walsh\}@data61.csiro.au}}

\maketitle

\begin{abstract}
Two simple and attractive mechanisms for
the fair division of indivisible goods in an online
setting are \mylike\ and \myblike.
We study some fundamental computational problems concerning the outcomes of these mechanisms. In particular, we consider what expected outcomes are \emph{possible}, what outcomes are \emph{necessary} and how to compute their \emph{exact} outcomes. In general, we show that
such questions are more tractable to compute for
\mylike\ than for \myblike. 
As \mylike\ is strategy proof but \myblike\ is
not, we also consider the computational problem
of how, with \myblike, 
an agent can compute a strategic bid to
improve their outcome. We prove
that this problem is intractable in general. 
\end{abstract}

\section{Introduction}\label{sec:intro}

Fair division is a fundamental problem in allocating resources among competing agents. Many practical fair division problems 
are online. We present two such settings. For example, in a food bank, we must start allocating food as it
is donated. It is too late to wait until the end of the day
before we start distributing the food to charities.
As a second example, in allocating deceased organs to patients
we must match newly donated organs swiftly. 
We cannot wait till more organs arrive before deciding on the
precise match.

Motivated by such problems, Walsh has proposed a simple online model for the fair division of indivisible items in which the items arrive over time \cite{walsh2015}. 
Aleksandrov et al. analysed two simple and attractive
randomized mechanisms
for such fair division problems: \mylike\ and \myblike\
\cite{aleksandrov2015ijcai}. 
The \mylike\ mechanism allocates
an arriving item uniformly at random between the
agents that ``like'' it. It satisfies equal treatment of equals, and
it is both strategy proof
and envy free ex ante 
\cite{aleksandrov2015ijcai}. Indeed, any mechanism that is envy free ex ante assigns items to agents with the same probabilities as \mylike\ does. However, the \mylike\ mechanism is not
very fair ex post as it can possibly allocate
all items to one agent. 
The \myblike\ mechanism 
is fairer. It allocates
an arriving item uniformly at random between the
agents that ``like'' it who have the fewest
items currently. 
\myblike\ bounds the envy one agent 
has for another's allocation ex post. However, this
comes at the price of no longer being 
strategy proof in general 
\cite{aleksandrov2015ijcai}.  
When restricted to 2 agents
and 0/1 utilities, \myblike\ is
strategy proof. These mechanisms are simple and satisfy many desirable axioms. For these reasons, we now turn attention to their computational properties. 

In practice, it may be difficult to
query the agents each time an item arrives. 
The chair will often collect
the preferences of the agents in advance,
and allocate items to agents as they arrive. There are several
settings where it is reasonable
to suppose that the chair does that. For instance, 
in the food bank problem, a good proxy
for the utility of an item to a charity
that likes it might simply be its
retail price. This is public information. 
As a second example, in 
deceased organ matching, the utility of allocating an organ to a patient might be computed from a simple formula that takes account
of the age of the organ, the age of
the patient and a number of other 
medical factors. This is again public information. The chair might then be interested in what outcomes 
are possible, necessary or exact based on these
declared preferences.
For example, the chair might
be concerned that agents 
receive enough utility or 
particular essential items.
Alternatively, the chair might
want to be sure that a favored
agent gets a particular item. Also, they might even want to give similar utility to each agent or bias the future allocation in case some agents receive only a few items and are promised to receive more in expectation.

There are two sources of uncertainty in 
deciding these outcomes. First,
both mechanisms are randomized.
Therefore each mechanism returns 
a probability distribution over actual outcomes. 
Second, as the problem is online, the arrival order of items 
is typically unknown. We consider here 
the problem of the chair computing what outcomes
are possible, necessary or exact depending on
both sources of uncertainty. In particular, we focus on computing whether 
an agent can possibly or necessarily 
receive a given expected utility. These results
easily translate into whether an agent
can possibly or necessarily 
receive a given item. We simply
give most of the agent's utility to that item. 
Also, as all our results hold in the case
of binary utilities, they can also be viewed
as computing whether an
agent can possibly or necessarily
receive a given expected number 
of items. Whilst some of our results consider general utilities, such utilities are mainly used to compare outcomes and do not need to be elicited explicitly. General utilities are not used when bidding or allocating items. Such ``like" and ``not like'' reporting has advantages. It is simple, does not require costly eliciting of utilities of agents for items and it also leads to mechanisms with nice
axioms. 

\emph{Our contributions:} We consider three settings:
the chair knows the arrival ordering of items,
the arrival ordering is drawn from some probability
distribution, and the allocation of past items is known. In all settings, we study the problem of the chair computing possible, necessary and exact outcomes of \mylike\ and \myblike. For both mechanisms, these problems are intractable
even with 2 agents and when the ordering of items is not fixed. In contrast, with any number of agents, computing each of these outcomes is tractable for \mylike\ and intractable for \myblike\ when the ordering of items is fixed. Interestingly, computing outcomes with \myblike\ becomes tractable in this setting only when restricted to 2 agents. Further, computing outcomes is tractable for both mechanisms at a certain moment of time when a new item arrives supposing the allocation of past items is known. In addition, we study a closely related problem of whether an agent can manipulate these mechanisms by strategically misreporting their preferences. Our computational results have a number of interesting consequences. For example, recall that the \myblike\ mechanism is fairer but not strategy proof. However, we show that computing a manipulation of this mechanism is intractable in general. 

\section{Preliminaries}\label{sec:prem} 

We next provide basic definitions of online instances, the \mylike\ and \myblike\ mechanisms and their outcomes.

{\bf Allocation instance:} An \emph{instance} $\mathcal{I}=(A,O,U,\Delta)$ of an online fair division problem has (1) a set $A$ of \emph{agents} $a_1,\ldots,a_n$, (2) a set $O$ of indivisible \emph{items} $o_1,\ldots,o_m$, (3) a matrix $U=(u_{ik})_{m\times n}$ where $u_{ik}$ is the {\em cardinal utility} of agent $a_i$ for item $o_k$ and (4) a matrix $\Delta=(\delta_{kj})_{m\times m}$ where $\delta_{kj}$ is a \emph{probability} that item $o_k$ arrives in moment $j$.

We consider \emph{binary} utilities and \emph{general} rational non-negative utilities. We say that agent $a_i$ \emph{likes} item $o_k$ if $u_{ik}>0$. Further, we assume that one item arrives in each moment $j$, i.e. $\sum_{k=1:m} \delta_{kj}=1$. 

{\bf Online setting:} Suppose items $o_{1}$ to $o_{j}$ have arrived at moments $1$ to $j$, respectively. Given $o=(o_{1},\ldots,o_{j})$, let $\Delta(o)$ be its probability, $\pi(j,o)$ the current allocation of these items to agents, $p(\pi(j,o))$ its probability and $u_i(\pi(j,o))$ the additive utility of agent $a_i$ for the items they receive in $\pi(j,o)$. Now, suppose that item $o_{k}$ arrives at moment $(j+1)$ with probability $\delta_{k(j+1)}$ when each agent $a_{i}$ places a rational non-negative \emph{bid} $v_{ik}$ for this item and a \emph{mechanism} then decides its allocation to a \emph{feasible} agent in an \emph{online} manner, i.g. given $\pi(j,o)$ and \emph{no} information about future items. 

{\bf Mechanisms:} We consider the randomized \mylike\ and \myblike\ mechanisms from \cite{aleksandrov2015ijcai}. With the \mylike\ mechanism, agent $a_i$ is feasible for item $o_k$ if $v_{ik}>0$. With the \myblike\ mechanism, agent $a_i$ is feasible for item $o_k$ if $v_{ik}>0$ and have so far received fewest items given $\pi(j,o)$ among those agents that bid positively for item $o_k$. Let the number of feasible agents be $f_k$. The probability that a feasible agent $a_i$ is allocated item $o_k$ is equal to $1/f_k$. 

{\bf Possible, necessary and exact outcomes:} We consider \emph{expected} probabilities depending on what information is available to the chair. If the allocation $\pi(j,o)$ is the only available information, we use $p_i(j+1,\pi(j,o))$ for the \emph{probability} of agent $a_i$ for the item that arrives at moment $(j+1)$. If the order $o$ is the only available information, we use $p_i(j+1,o)$ for the \emph{probability} of agent $a_i$ for the item that arrives at moment $(j+1)$. It is equal to $\sum_{\pi(j,o)}p(\pi(j,o))\cdot p_i(j+1,\pi(j,o))$. If there is no information about $o$ or $\pi(j,o)$, we use $p_i(j+1)$ for the \emph{probability} of agent $a_i$ for the item that arrives at moment $(j+1)$. It is equal to $\sum_{o}\Delta(o)\cdot p_i(j+1,o)$. We next define expected utilities of agents for items in each of these settings. Given $\pi(j,o)$, we use $u_{ij}(\pi(j,o))$ for the \emph{utility} of agent $a_i$. It is equal to $u_i(\pi(j,o))+p_i(j+1,\pi(j,o))$. Given $o$, we use $u_{ij}(o)$ for the \emph{utility} of agent $a_i$. It is equal to $\sum_{\pi(j,o)}p(\pi(j,o))\cdot u_i(\pi(j,o))$. Given $\Delta$, we use $u_{ij}(\Delta)$ for the \emph{utility} of agent $a_i$. It is equal to $\sum_{o}\Delta(o)\cdot u_{ij}(o)$.

The \emph{probability (or utility)} of agent $a_i$ at moment $j$ is \emph{possible} if their probability (or utility) is positive. The outcome of agent $a_i$ at moment $j$ is \emph{necessary} at least some rational number $k$ if their probability (or utility) is at least $k$. We also say that the outcome of agent $a_i$ at moment $j$ is \emph{exact} if we want to compute the exact value of their probability (or utility). 

We study the complexity of computing possible, necessary and exact outcomes. For a mechanism that allocates all items to agents that like them, note that possible and necessary outcomes are directly related. For this reason, we only study necessary and exact outcomes. Our results for possible outcomes are inherited. We next show this relation. 

Suppose we ask if ${p}_{i}(j+1)>0$ holds. This is true iff there is an ordering $o$ and allocation $\pi(j,o)$ of the first $j$ items such that $p_i(j+1,\pi(j,o))>0$. We therefore conclude that ${p}_{i}(j+1)>0$ iff ${p}_{i}(j+1)\geq \epsilon$ where $0<\epsilon\leq\min_{o,\pi(j,o)}\Delta(o)\cdot p(\pi(j,o))\cdot p_i(j+1,\pi(j,o))$. Note that this minimum value is positive and, consequently, such $\epsilon$ always exists. Such a relation is not true for utilities. For the utility of agent $a_i$, we have that $u_{ij}(\Delta)>0$ holds iff agent $a_i$ bids positively for at least one item and at least one item arrives. This problem is easy to decide. However, deciding if $u_{ij}(\Delta)\geq k$ holds might not be so easy. 

Recall that we consider three settings: when the past allocation of items to agents is known, when the ordering of items is unknown and when the ordering of items is known. We next observe that all outcomes are tractable in the setting when the past allocation is known, \emph{fixed} and \emph{no} information about future items is available. 

{\bf Items arriving online:} Let us suppose that the first $j$ items have arrived and their allocation be $\pi(j,o)$. Suppose now that item $o_{k}$ arrives at moment $(j+1)$. For both \mylike\ and \myblike, the \emph{exact} value of $p_i(j+1,\pi(j,o))$ is equal to $\sum_{k=1:m}\delta_{k(j+1)}\cdot(1/f_k)$ and the \emph{exact} value of $u_i(\pi(j,o))$ is equal to the sum of the cardinal utilities of agent $a_i$ for the items they are allocated in $\pi(j,o)$. Both of these exact outcomes, the value of $u_{ij}(\pi(j,o))$ and therefore any \emph{possible} and \emph{necessary} outcomes in this setting can be computed in $\mathcal{O}(m\cdot n)$ time and space.

We use popular reductions and computational problems from computational complexity, graph theory and set theory in order to show our hardness results. 

{\bf Computational complexity:} We use complexity classes of decision and counting problems such as $\P$, $\NP$, $\coNP$ and $\#\P$, and mappings such as \emph{Karp}, \emph{Turing}, \emph{parsimonious} and \emph{arithmetic} reductions \cite{burgisser2000,turing1936a,valiant1979p}. 

{\bf Graph theory:} Let $G$ be an undirected bipartite graph. A \emph{matching} $\mu$ in $G$ is a set of vertex-disjoint edges. We say that $\mu$ \emph{matches} a vertex if there is an edge in it that is incident with the vertex. Matching $\mu$ is \emph{maximal} if it is no longer a matching once some other edge is added to it.  Matching $\mu$ is \emph{perfect} if it matches all vertices in $G$. Given a graph $G$ and a number $k$, the \emph{minimum size maximal matching problem} is to decide if there is a matching $\mu$ in $G$ with $|\mu|\leq k$. It is $\NP$-hard on various bipartite graphs \cite{demange2008,saban2015}. Given a graph $G$, the \emph{counting perfect matchings problem} is to output the number of perfect matchings in $G$. It is $\#\P$-hard on various bipartite graphs \cite{okamoto2009,valiant1979er}.

{\bf Set theory:} Let $S$ be a set of integers and $b,c$ be integers. A $(b,c)$-subset of $S$ is a subset of $S$ whose elements sum up to $b$ and its cardinality is $c$. The \emph{$(b,c)$-subset sum problem} is to decide if there is a $(b,c)$-subset of $S$. Note that there is a $(b,c)$-subset of $S$ for at least one $c\in[1,|S|]$ iff there is a subset of $S$ whose elements sum up to $b$. The latter problem is the $\NP$-hard \emph{$b$-subset sum problem} \cite{garey1979}. 

This paper is structured as follows. In Section~\ref{sec:second}, the items are drawn from some known probabilistic distribution $\Delta$. For example, such distribution in the food bank problem could be estimated based on historical data. In Section~\ref{sec:third}, we suppose the ordering $o$ in which the items will arrive is fixed, i.e. for each moment $j$, we have that $\delta_{kj}=1$ holds for exactly one item $o_k$. Again, in the food bank problem, some charities donate certain items on a regular basis and only at specific moments. In Section~\ref{sec:manipulations}, we consider problems of computing manipulations of these mechanisms.

\section{Items Arriving from a Distribution}\label{sec:second}

We suppose the agents act sincerely and begin with the case when the chair knows the utilities {\em but} the items come from a distribution $\Delta$ whose size is polynomial in $n$ and $m$.

\begin{center}
\begin{tabular}{ccc}

\fbox{
\begin{minipage}{0.45\textwidth}
\mystochastic\myexact\myutility \\
Input: $\mathcal{I}=(A,O,U,\Delta)$, $a_i$. \\
Output: ${u}_{im}(\Delta)$.
\end{minipage}
}

& \hspace{0.1cm} &

\fbox{
\begin{minipage}{0.45\textwidth}
\mystochastic\mynecessary\myutility ~ \\
Input: $\mathcal{I}=(A,O,U,\Delta)$, $a_i$, $k\in\mathbb{Q}$.
\\
Question: ${u}_{im}(\Delta)\geq k$?
\end{minipage}
}

\end{tabular}
\end{center}

The stochastic exact outcomes of \mylike\ and \myblike\ are $\#\P$-hard with just two agents. Our reduction is motivated by the food bank problem. Let $m$ items be donated by $m$ suppliers and not each of the suppliers can donate each of the items. This relation could be viewed as an undirected bipartite graph. The items are in one partition. The suppliers are in another partition. Let us enumerate them from $1$ to $m$. There is an edge between an item and a supplier if the supplier donates the item. Each perfect matching in the graph then can be viewed as an ordering w.r.t. the enumeration of the suppliers in which each of the $m$ different suppliers donates exactly one of the $m$ different items. At the beginning of the day, the chair does not know the actual order in which the suppliers will donate items but they can estimate it by computing an estimate $\delta_{kj}$ for each item $o_k$ and moment $j$. Based on past data whose size is polynomial in $m$, one such estimate could be the number of days of past data in which each of the $m$ items is donated from a different supplier amongst the $m$ suppliers divided by the total number of days of past data. We give a reduction from the \emph{counting perfect matchings problem} to \mystochastic\myexact\myutility.

\begin{myreduction}\label{red:one}\em
Let $G$ be a (3-regular) bipartite graph with $M$ vertices in each partition. The allocation instance $\mathcal{I}_G$ has:

\begin{itemize}[noitemsep,topsep=0pt]
\item {\bf Agents:} agents $a_1$ and $a_2$ (i.e. $2$ agents),
\item {\bf Items:} items $o_1$ to $o_M$ (i.e. $M$ items), 
\item {\bf Utilities:}  $u_{ij}=1$ for each $a_i$ and $o_j$, and
\item {\bf Distribution:} $\delta_{kj}=1/M$ for each $o_k$ and $j$.
\end{itemize}

\end{myreduction}

\begin{mytheorem}\label{thm:one} 
With $n=2$ agents, 0/1 utilities and the \mylike\ or \myblike\ mechanism, problem {\sc StochasticExactUtility} is $\#\P$-hard under arithmetic reductions. 
\end{mytheorem}

\begin{proof}
WLOG, the set of orderings of items is equal to the set of perfect matchings in $G$ united with the set of $o_{\epsilon}$ that reveals no items. Each ordering $o_M$ that reveals $M$ items corresponds to a perfect matching in $G$ w.r.t. the enumeration of the suppliers in $G$. We suppose the items arrive independently of each other and across the different time moments. Consequently, ordering $o_M$ occurs with probability $1/M^M$ and the expected utility $u_{iM}(o_M)$ is $M/2$ with both mechanisms as both agents have the same utilities for items. The ordering $o_{\epsilon}$ reveals 0 items. It occurs with probability $1$ minus $(1/M^M)$ multiplied by the number of perfect matchings in $G$ and $u_{i0}(o_{\epsilon})$ is $0$ with both mechanisms as no items are revealed. We quickly obtain that $u_{iM}(\Delta)$ is equal to $(1/M^M)\cdot(M/2)$ multiplied by the number of perfect matchings in $G$. The result follows.\hfill \myqed
\end{proof}

We further showed that stochastic necessary outcomes of these mechanisms are $\NP$-hard with just two agents. We omit the complete proof for reasons of space but we give the main reduction which is from the {\em $(b,c)$-subset sum problem}. Given set of integers $S=\lbrace n_1,\ldots,n_M\rbrace$ and integers $b$ and $c$, we construct instance $\mathcal{I}_{S,b,c}$: (1) agents $a_1$ and $a_2$, (2) item $o_k$ for each $n_k\in S$, (3) agent $a_i$ values item $o_k$ with $n_k$, and (4) $\delta_{kj}=1/M$ for each item $o_k$ and moment $j$. The instance of {\sc StochasticNecessaryUtility} has $\mathcal{I}_{S,b,c}$, agent $a_i$ and constant $k=(1/M^c)\cdot(b/2)$. Let us order each subset of $S$ w.r.t. the enumeration $(1,\ldots,M)$. The set of orderings is now equal to the set of ordered $(b,c)$-subsets of $S$ united with the set of $o_{\epsilon}$ that reveals no items. Similarly to the proof of Theorem~\ref{thm:one}, it should be easy now for the reader to show that there is a $(b,c)$-subset of $S$ iff $u_{iM}(\Delta)\geq k$.

\section{Items Arriving from a Fixed Ordering}\label{sec:third}

We again suppose the agents act sincerely and next consider the case that the chair knows the utilities {\em and} the arrival ordering of future items. This corresponds to the case when exactly one item arrives with probability of one at each moment in time.

\begin{center}
\begin{tabular}{ccc}

\fbox{
\begin{minipage}{0.45\textwidth}
 \myexact\myutility \\
Input: $\mathcal{I}=(A,O,U,o)$, $a_i$. \\
Output: ${u}_{im}(o)$.
\end{minipage}
}

& \hspace{0.1cm} &

\fbox{
\begin{minipage}{0.45\textwidth}
\mynecessary\myutility ~ \\
Input: $\mathcal{I}=(A,O,U,o)$, $a_i$, $k\in\mathbb{Q}$.
\\
Question: ${u}_{im}(o)\geq k$?
\end{minipage}
}

\end{tabular}
\end{center}

\subsection{The Case of $n>2$ Agents}

Let there be $n>2$ agents. Interestingly, the outcomes of the \mylike\ mechanism become tractable whereas the ones of the \myblike\ mechanism remain intractable even when the ordering is fixed.

\subsubsection{Exact Outcomes}

Let us start with the \mylike\ mechanism. This mechanism does not keep track of the allocation of past items. As a result, any agent is feasible for each next item supposing they like this item. Indeed, all exact outcomes are tractable with this mechanism for this reason.
 
\begin{myobservation}\label{obs:one} 
With general utilities and the \mylike\ mechanism, problem {\sc ExactUtility} is in $\P$. 
\end{myobservation}

\begin{proof}
The probability $p_i(j,o)$ of agent $a_i$ for item $o_j$ is $1/n_j$ where $n_j$ is the number of agents that like the item. Their utility $u_{im}(o)$ can be given as $\sum_{j=1}^m (1/n_j)\cdot u_{ij}$. \hfill \myqed
\end{proof}

We continue with exact allocations for the \myblike\ mechanism and give a parsimonious reduction from \emph{counting perfect matchings problem} to \myexact\myutility. The counting problem remains in $\#\P$-hard even on $3$-regular undirected bipartite graphs in \cite{dagum1992}. Our reduction is very insightful because it provides a very tight bound on the complexity of \myexact\myutility\ (i.e. 0/1 utilities, each agent likes at most 4 items, each item except one is liked by at most 3 agents, each pair of agents like at most 3 items in common, the ordering is fixed, etc.).

\begin{myreduction}\label{red:two}\em
Let $G$ be a 3-regular bipartite graph, $u_1,\ldots,u_N$ be the vertices from one of its partitions and $v_1,\ldots,v_N$ the vertices from the other one of its partitions. For each vertex $u_i$, let $v_{i1},v_{i2},v_{i3}$ denote the vertices connected to it and $e_{3\cdot(i-1)+1}=(u_i,v_{i1})$, $e_{3\cdot(i-1)+2}=(u_i,v_{i2})$, $e_{3\cdot(i-1)+3}=(u_i,v_{i3})$ the edges incident with it. Each edge $e_k$ can be represented as $(u_i,v_j)$ for some $u_i\in\lbrace u_1,\ldots,u_N\rbrace$ and $v_j\in\lbrace v_{i1},v_{i2},v_{i3}\rbrace$. We use the graph and next construct the online allocation instance $\mathcal{E}_{G}$ as follows: 

\begin{itemize}[noitemsep]
\item {\bf Agents:} 1 agent $a_k$ per edge $e_k$ and 3 special agents $a_{3\cdot N+1}$, $a_{3\cdot N+2}$ and $a_{3\cdot N+3}$ (i.e. $3\cdot N+1$ agents),
\item {\bf Items:} 1 item per vertex $v_j$, 2 items $u_{i1}$, $u_{i2}$ per vertex $u_i$ and 3 special items $w$ and $x$ (i.e. $3\cdot N+2$ items), 
\item {\bf Non-zero utilities:} for $i\in[1,N],j\in\lbrace 1,2,3\rbrace$, {\bf agent} $a_{3\cdot(i-1)+j}$ has utility 1 for items $v_{ij},u_{i1},u_{i2},x$; {\bf agent} $a_{3\cdot N+1}$ has utility 1 for items $w,x$, and
\item {\bf Ordering:} $o=(v_1\ldots v_Nu_{11}u_{12}\ldots u_{N1}u_{N2}wx)$.
\end{itemize}

\end{myreduction}

We highlight the main idea behind the proof of the next Lemma~\ref{lem:one}. Basically, we showed that computing the number of allocations of the first $3\cdot N+1$ items in $o$ in which each agent receives exactly one item is in $\#\P$-complete.

\begin{mylemma}\label{lem:one}  With the \myblike\ mechanism, the number of allocations in $\mathcal{E}_{G}$ in which agent $a_{3\cdot N+1}$ is feasible for item $x$ is equal to $2^N$ times the number of perfect matchings in $G$. Computing it is in $\#\P$-hard under arithmetic reductions.
\end{mylemma}

\begin{proof} By construction, each item $v_j$ is liked by three different agents and, hence, each allocation of $v_1,\ldots,v_N$ gives these items to $N$ different agents among $a_1,\ldots,a_{3\cdot N}$. Consider then an allocation of $v_1,\ldots,v_N$ such that, for each vertex $u_i$, either agent $a_{3\cdot(i-1)+1}$ gets item $v_{i1}$ or agent $a_{3\cdot(i-1)+2}$ gets item $v_{i2}$ or agent $a_{3\cdot(i-1)+3}$ gets item $v_{i3}$. We say that such an allocation of $v_1,\ldots,v_N$ has \emph{perfect} matches for vertices $u_1,\ldots,u_N$ because exactly one agent per triplet $a_{3\cdot(i-1)+1},a_{3\cdot(i-1)+2},a_{3\cdot(i-1)+3}$ gets an item among $v_1,\ldots,v_N$. In fact, there is a perfect matching in $G$ over $v_1,\ldots,v_N$ and $u_1,\ldots,u_N$ iff there is an allocation in $\mathcal{E}_{G}$ of $v_1,\ldots,v_N$ that has perfect matches for $u_1,\ldots,u_N$. Furthermore, this is a 1-to-1 parsimonious correspondence. Each allocation $\pi$ in $\mathcal{E}_{G}$ of the first $3\cdot N+1$ items in $o$ in which each agent among $a_1,\ldots,a_{3\cdot N},a_{3\cdot N+1}$ receives exactly one item occurs with positive probability. We call $\pi$ \emph{perfect allocation} over the first $3\cdot N+1$ items in $o$. We show that there is an allocation in $\mathcal{E}_{G}$ of $v_1,\ldots,v_N$ that has perfect matches for $u_1,\ldots,u_N$ iff there are $2^N$ perfect allocations such as $\pi$ in $\mathcal{E}_{G}$. Moreover, this is a 1-to-$2^N$ arithmetic correspondence. In other words, we show that the number of perfect allocations such as $\pi$ in $\mathcal{E}_{G}$ is equal to $2^N$ times the number of perfect matchings in $G$.

First, let us consider one discrete allocation $\pi_1$ in $\mathcal{E}_{G}$ of $v_1,\ldots,v_N$ that has perfect matches for $u_1,\ldots,u_N$. The allocation $\pi_1$ occurs with positive probability because $v_1,\ldots,v_N$ are liked by disjoint sets of three agents. WLOG, suppose that $\pi_1$ is such that, for each $u_i$, agent $a_{3\cdot(i-1)+1}$ receives their corresponding item $v_{i1}$. The allocation $\pi_1$ can be extended by the mechanism to two discrete allocations w.r.t. each $u_i$: (1) agent $a_{3\cdot(i-1)+2}$ gets item $u_{i1}$ and agent $a_{3\cdot(i-1)+3}$ gets item $u_{i2}$ or (2) agent $a_{3\cdot(i-1)+2}$ gets item $u_{i2}$ and agent $a_{3\cdot(i-1)+3}$ gets item $u_{i1}$. By the preference structure, $\pi_1$ can then be extended by the mechanism to $2^N$ perfect allocations in $\mathcal{E}_{G}$. Note that each of these perfect allocations necessarily gives item $w$ to agent $a_{3\cdot N+1}$ because only they like it. Second, consider one perfect allocation in $\mathcal{E}_{G}$. It must be the case that it extends some discrete allocation of $v_1,\ldots,v_N$ that has perfect matches for $u_1,\ldots,u_N$. To show this, consider a discrete allocation $\pi_2$ of $v_1,\ldots,v_N$ that has not perfect matches for $u_1,\ldots,u_N$. Hence, $\pi_2$ is such that at least two of the agents $a_{3\cdot(i-1)+1},a_{3\cdot(i-1)+2},a_{3\cdot(i-1)+3}$ for some vertex $u_i$ receive their corresponding items $v_{i1},v_{i2},v_{i3}$ of $v_1,\ldots,v_N$. Therefore, each allocation of all items that extends $\pi_2$ by using the mechanism gives item $u_{i1}$ or item $u_{i2}$ to one of the agents $a_{3\cdot(i-1)+1},a_{3\cdot(i-1)+2},$ $a_{3\cdot(i-1)+3}$ as their second item. As a consequence, in each such allocation, there is another agent with zero items after round $3\cdot N+1$. We conclude that each such extension of $\pi_2$ is not a perfect allocation in $\mathcal{E}_{G}$. \hfill \myqed
\end{proof}

\begin{mytheorem}\label{thm:two} 
With $n>2$ agents, 0/1 utilities and the \myblike\ mechanism, problem {\sc ExactUtility} is in $\#\P$-hard under arithmetic reductions. 
\end{mytheorem}

\begin{proof} Let us consider allocation $\pi=\pi(3\cdot N+1,o)$ of the first $3\cdot N+1$ items in $o$ in which each agent among $a_1,\ldots,a_{3\cdot N},a_{3\cdot N+1}$ receives exactly one item. Note that agent $a_{3\cdot N+1}$ gets item $x$ with positive conditional probability only given such allocations because all agents like item $x$. By the preference structure, we conclude that $\pi$ occurs with probability $p(\pi)=(1/3^N)\cdot(1/2^N)$. The conditional probability $p_i(x|\pi)$ of agent $a_{3\cdot N+1}$ for item $x$ given $\pi$ is equal to $1/(3\cdot N+1)$ because all agents $a_1,\ldots,a_{3\cdot N},a_{3\cdot N+1}$ like item $x$. The conditional probability of agent $a_{3\cdot N+1}$ for item $x$ is $0$ given any other allocation. Therefore, $p_{3\cdot N+1}(x,o)$ is equal to $(1/3^N)\cdot(1/2^N)\cdot(1/(3\cdot N+1))$ multiplied by the number of allocations such as $\pi$ in which agent $a_{3\cdot N+1}$ is feasible for item $x$. Finally, the expected utility ${u}_{(3\cdot N+1)(3\cdot N+3)}(o)=p_{3\cdot N+1}(w,o)+p_{3\cdot N+1}(x,o)$. We have that $p_{3\cdot N+1}(w,o)=1$ because only agent $a_{3\cdot N+1}$ likes item $w$ and the mechanism allocates each item to an agent. The result follows by Lemma~\ref{lem:one}.\hfill \myqed
\end{proof}

\subsubsection{Necessary Outcomes}

The tractability of the exact allocations of the \mylike\ mechanism entails the tractability of its necessary allocations. By Observation~\ref{obs:one}, we conclude the next immediate result.

\begin{myobservation}\label{obs:two} 
With general utilities and the \mylike\ mechanism, problem {\sc NecessaryUtility} is in $\P$. 
\end{myobservation}

We next focus on the necessary outcomes of the \myblike\ mechanism. We give a Karp reduction from \emph{minimum size maximal matching problem} to the negation of {\sc NecessaryUtility}. The minimum size maximal matching problem is shown to be $\NP$-hard on subdivision graphs of degree at most 3 in \cite{horton1993}.

\begin{myreduction}\label{red:three}\em Let us have a subdivision graph $G$ of degree at most 3 and integer $r$. The graph $G$ is bipartite with vertices $u_1,\ldots,u_N$ of degree exactly 2 and vertices $v_1,\ldots,v_M$ of degree at most 3. WLOG, we can assume that $N\geq M$ and there are no two vertices from $U$ that are connected to the same two vertices from $V$. We construct an allocation instance $\mathcal{P}_{G,r}$ as follows:

\begin{itemize}[noitemsep]
\item {\bf Agents:} 2 agents $u_{i1}$, $u_{i2}$ per $u_i$ and agents $a_1,\ldots,a_{N-r}$, $b_1,\ldots,b_M$ and $c$ (i.e. $3\cdot N+M-r+1$ agents),
\item {\bf Items:} 1 item per $v_j$ and items $x_1,\ldots,x_N$, $y_1,\ldots,y_N$, $z_1,\ldots,z_{N-r}$ and $w$ (i.e. $3\cdot N+M-r+1$ items), 
\item {\bf Non-zero utilities:} for each $i\in[1,N],j\in\lbrace 1,2\rbrace$, {\bf agent} $u_{ij}$ has utility 1 for items $x_i,v_{ij},y_i,z_1,\ldots,z_{N-r}$; for each $i\in[1,N-r]$, {\bf agent} $a_i$ has utility 1 for items $x_1,\ldots,x_N$; {\bf agents} $b_1,\ldots,b_M$ have each utility 1 for item $w$; {\bf agent} $c$ has utility 1 for items $z_{N-r},w$, and
\item {\bf Ordering:} $o=(x_1\ldots x_Nv_1\ldots v_My_1\ldots y_Nz_1\ldots z_{N-r}w)$. 
\end{itemize}

\end{myreduction}

The expected utility of each of the agents $b_1,\ldots,b_M$ is at least $1/M$ iff $p_{c}(w,o)=0$. This observation holds because each of the agents $b_1$ to $b_M$ have equal utilities for items in which case they receive item $w$ with the same probability which apparently is also equal to their expected utility as this is the only item they like. Theorem~\ref{thm:three} follows from this observation.

\begin{mytheorem}\label{thm:three} 
With $n>2$ agents, 0/1 utilities and the \myblike\ mechanism, problem {\sc NecessaryUtility} is in $\coNP$-hard under Turing reductions. 
\end{mytheorem}

\begin{proof}
There is a maximal matching in $G$ of cardinality at most $r$ iff there is an allocation in $\mathcal{P}_{G,r}$ in which agent $c$ receives item $w$ iff $p_{c}(w,o)>0$. The second ``iff'' is trivial. We, therefore, focus on the first ``iff''. The ``only if'' direction is easier to show and, for reasons of space, we only show the more difficult ``if'' direction. Suppose next that $\pi$ is an allocation of all items in $\mathcal{P}_{G,r}$ in which agent $c$ receives item $w$. 

\begin{enumerate}[noitemsep,topsep=0pt]
\item Item $w$ is allocated in $\pi$ to agent $c$ as their first item. To see this, suppose they also get some items among $z_{N-r}$. Now, they would not be feasible when item $w$ arrives as agents $b_1,\ldots,b_M$ have zero items in $\pi$ and the mechanism would have given item $w$ to an agent among $b_1,\ldots,b_M$ and not to agent $c$.

\item Prior to item $w$ in $\pi$, agent $c$ have received zero items. Hence, items $z_1,\ldots,z_{N-r}$ are allocated in $\pi$ to $N-r$ agents as their first items. By the preferences, these agents are from different pairs among $u_{11},u_{12},\ldots,u_{N1},$ $u_{N2}$ because, for each pair of agents $u_{i1},u_{i2}$, either $u_{i1}$ or $u_{i2}$ is forced to get item $y_i$. WLOG, let us assume that agents $u_{11},\ldots,u_{(N-r)1}$ get items $z_1,\ldots,z_{N-r}$ in $\pi$.

\item Prior to item $z_1$ in $\pi$, agents $u_{11},\ldots,u_{(N-r)1}$ have zero items. Hence, $N-r$ items among $y_1,\ldots,y_N$ are allocated in $\pi$ to $u_{12},\ldots,u_{(N-r)2}$ as their first items. These items are $y_1,\ldots,y_{N-r}$. For $i$ in $[N-r+1,N]$, we note that item $y_i$ is allocated in $\pi$ to either $u_{i1}$ or $u_{i2}$ as their first or second item.

\item Prior to item $y_1$ in $\pi$, agents $u_{11},u_{12},\ldots,u_{(N-r)1},$ $u_{(N-r)2}$ have zero items. By the preferences, agents $a_1,\ldots,$ $a_{N-r}$ must then receive items $x_1,\ldots,x_{N-r}$ in $\pi$. For $i$ in $[N-r+1,N]$, item $x_i$ is allocated in $\pi$ to either $u_{i1}$ or $u_{i2}$, say $u_{i2}$. We conclude that agents $u_{(N-r+1)1},\ldots,u_{N1}$ have zero items prior to item $v_1$ in $\pi$. Moreover, only agents $u_{(N-r+1)1},u_{(N-r+1)2},\ldots,$ $u_{N1},u_{N2}$ receive items $v_1,\ldots,v_M$ in $\pi$. Finally, only $l\leq r$ agents among $u_{(N-r+1)1},\ldots,u_{N1}$ get items in $\pi$ among $v_1,\ldots,v_M$ as first items as some of these agents might like the same items among $v_1,\ldots,v_M$. WLOG, let these agents be $u_{(N-l+1)1},\ldots,u_{N1}$ and they are allocated in $\pi$ items $v_1,\ldots,v_l$ as first items.
\end{enumerate} 

The constructed set $\mu_{\pi}=\lbrace (u_{N-l+1},v_1),\ldots,(u_N,v_l)\rbrace$ contains only edges from the graph $G$ which are vertex-disjoint. Therefore, this set is a matching in $G$. Moreover, the cardinality of this set is $l$ at most $r$. We next show that $\mu_{\pi}$ is a maximal matching. For the sake of contradiction, suppose that $\mu_{\pi}$ remains a matching if we add a new edge to it, say $(u,v)$. The edge $(u,v)$ is vertex-disjoint with the edges in $\mu_{\pi}$. This means that vertex $u$ is not among $u_{N-l+1},\ldots,u_N$ and vertex $v$ is not among $v_1,\ldots,v_l$. Hence, vertex $u$ is among $u_{1},\ldots,u_{N-l}$. In the allocation $\pi$, agents $u_{11},u_{12},\ldots,u_{(N-r)1},u_{(N-r)2}$ do not receive any items among $v_1,\ldots,v_M$. This implies that all these agents are feasible for the items they like among $v_1,\ldots,v_M$ but they do not get them in $\pi$. As agents $u_{(N-l+1)1},\ldots,u_{N1}$ get items $v_1,\ldots,v_l$ as their first items, we conclude that some agents among $u_{(N-l+1)1},u_{(N-l+1)2}\ldots,u_{N1},u_{N2}$ receive items $v_{l+1},$ $\ldots,v_M$ as their second items. Therefore, it must be the case that all agents $u_{11},u_{12},\ldots,u_{(N-r)1},u_{(N-r)2}$ do not like any item among $v_{l+1},\ldots,v_M$. Otherwise, the mechanism would allocate some of these items to agents among $u_{11},u_{12},$ $\ldots,u_{(N-r)1},u_{(N-r)2}$. This is just the way in which the mechanism works. And, we reached a contradiction with the existence of the allocation $\pi$. Finally, in the graph $G$, vertices $u_{1},\ldots,u_{N-r}$ are connected only to vertices among $v_1,\ldots,v_l$. Hence, $v$ is among $v_1,\ldots,v_l$. This fact contradicts that $\mu_{\pi}\cup\lbrace (u,v)\rbrace$ is a matching. \hfill \myqed
\end{proof}

\subsection{The Case of $2$ Agents}

By Observations~\ref{obs:one} and~\ref{obs:two}, the outcomes of \mylike\ are tractable. Surprisingly, in contrast to Theorems~\ref{thm:one},~\ref{thm:two} and~\ref{thm:three}, the outcomes of \myblike\ become tractable with only two agents and when the ordering of items is fixed.

\begin{mytheorem}\label{thm:four} 
With $n=2$ agents, general utilities and the \myblike\ mechanism, problems {\sc ExactUtility} and {\sc NecessaryUtility} are in $\P$. 
\end{mytheorem}

\begin{proof}
We use a dynamic program. Each state $s=(p,q)$ in it encodes that agent $a_1$ has $p$ items, agent $a_2$ has $q$ items, and its probability $p(s)$. By induction, we show that there are at most 2 different states after each allocation round. In the base case, consider round $1$. There are at most 2 states after this round depending on whether both $a_1$ and $a_2$ or only one of them like the first item. In the hypothesis, consider round $j$ and suppose there are at most two states after round $j$. In the step case, consider round $j+1$. Now, there are two cases. In the first one, there is only one state after round $j$. The result follows by the base case. In the second case, there are two states after round $j$. Let these be $(p,q)$ and $(p-1,q+1)$ where $p+q=j$. If only one agent likes item $o_{j+1}$, each state transits into a new state and the result follows. If both $a_1$ and $a_2$ like item $o_{j+1}$, we consider four sub-cases depending on the difference $p-q$: (1) $(p,q)$ and $(p-1,q+1)$ for $p-q>2$, (2) $(q+2,q)$ and $(q+1,q+1)$ for $p-q=2$, (3) $(q+1,q)$ and $(q,q+1)$ for $p-q=1$ and (4) $(q,q)$ and $(q-1,q+1)$ for $p-q=0$. For sub-case (1), each state transits into one new state with the same probability. For sub-case (2), $(q+2,q)$ transits into $(q+2,q+1)$, and $(q+1,q+1)$ into $(q+2,q+1)$ and $(q+1,q+2)$. For sub-case (3), both states transit into the same new state with probability $1$. For sub-case (4), $(q,q)$ transits into $(q,q+1)$ and $(q+1,q)$, and $(q-1,q+1)$ into $(q,q+1)$. We conclude that there are at most two different states after round $j+1$ in each sub-case. 

The probability $p_{1}(j+2,o)$ is equal to $\sum_{s_{j+1}} p(s_{j+1})\cdot p(a_1$ gets $o_{j+2}|s_{j+1})$ where $s_{j+1}$ is such a state after round $j+1$ in which agent $a_1$ is feasible for item $o_{j+2}$. The conditional probability $p(a_1$ gets $o_{j+2}|s_{j+1})$ of agent $a_1$ for item $o_{j+2}$ is ($\mathrm{\romannumeral 1}$) $0$ or $1$ in sub-case (1), ($\mathrm{\romannumeral 2}$) $0$, $1/2$ or $1$ in sub-case (3) and ($\mathrm{\romannumeral 3}$) the probability of the state in which they are feasible in sub-cases (2) and (4). We can compute the states, their probabilities and hence the probabilities of agents and their utilities in $\mathcal{O}(m)$ space and time. \hfill \myqed
\end{proof}

\section{Manipulations}\label{sec:manipulations}

We next consider how agents can act strategically. The \mylike\ mechanism is strategy-proof and hence agents 
have an incentive to bid sincerely for items. In contrast, the \myblike\ mechanism is not strategy-proof and agents can
have an incentive to bid strategically for items \cite{aleksandrov2015ijcai}. We thus focus on strategic misreporting of bids with \myblike. In particular, we study the worst case when the utilities and the ordering of the items are known to the misreporting agent. Any complexity results, in this case, provide lower
bounds on the complexity in the case of partial
or probabilistic information. We formulate the next problems where ${u}_{im}({v}^i,o)$ denotes the utility of agent $a_i$ supposing their bid vector is ${v}^i=(v_{i1},\ldots,v_{im})$ and the other agents bid sincerely. Let ${u}^i=(u_{i1},\ldots,u_{im})$ denotes their sincere bid vector.

\begin{center}
\begin{tabular}{ccc}

\fbox{
\begin{minipage}{0.45\textwidth}
 {\sc ExactManipulation} \\
Input: $\mathcal{I}=(A,O,U,o)$, $a_i$, ${u}^i$, ${v}^i$. \\
Output: ${u}_{im}({v}^i,o)-{u}_{im}({u}^i,o)$.
\end{minipage}
}

& \hspace{0.1cm} &

\fbox{
\begin{minipage}{0.45\textwidth}
 {\sc NecessaryManipulation} \\
Input: $\mathcal{I}$=$(A,O,U,o)$, $a_i$, ${v}^i$, ${u}^i$, $k$ $\in$ $\mathbb{Q}$. \\
Question: ${u}_{im}({v}^i,o)-{u}_{im}({u}^i,o)\geq k$?
\end{minipage}
}

\end{tabular}
\end{center}

\begin{mytheorem}\label{thm:five} 
With $n>2$ agents, 0/1 utilities and the \myblike\ mechanism, problem {\sc ExactManipulation} is in $\#\P$-hard under arithmetic reductions. 
\end{mytheorem}

\begin{proof}
Consider instance $\mathcal{E}_{G}$. Let us modify this instance a bit. We add one new item $z$ between items $w$ and $x$ in the ordering $o$ such that only agent $a_{3\cdot N+1}$ likes $z$ with 1. Let $\mathcal{F}_G$ denote this new instance. Suppose that all agents in $\mathcal{F}_G$ bid sincerely. Thus, agent $a_{3\cdot N+1}$ receives each of the items $w$ and $z$ each with probability 1 because they are the only agent who likes them. However, they receive item $x$ with probability 0. Therefore, ${u}_{(3\cdot N+1)(3\cdot N+3)}({u}^{(3\cdot N+1)},o)=2$. Suppose that all agents in $\mathcal{F}_G$ bid sincerely except agent $a_{3\cdot N+1}$ who bids strategically 0 for item $z$. Let ${v}^{(3\cdot N+1)}$ be their bidding vector in this case. We can now remove item $z$ because no agent bids positively for it. But, then we obtain instance $\mathcal{E}_G$. By Theorem~\ref{thm:two}, we have ${u}_{(3\cdot N+1)(3\cdot N+3)}({v}^{(3\cdot N+1)},o)=1+p_{3\cdot N+1}(x,o)$. The instance of {\sc ExactManipulation} uses instance $\mathcal{F}_G$, agent $a_{3\cdot N+1}$ and vectors ${u}^{(3\cdot N+1)}$ and ${v}^{(3\cdot N+1)}$. Its hardness follows by Theorem~\ref{thm:two}. \hfill \myqed
\end{proof}

Observe that the truthful report of agent $a_{3\cdot N+1}$ in the proof of Theorem~\ref{thm:five} leads to their utility being $2$ whereas their insincere report leads to their utility being at most $2$. Hence, their strategic move cannot lead to an increase in their utility but the computation of the exact difference in utility is intractable. However, as we discuss next, computing an exact profitable insincere report that leads to such an increase is also intractable.

Necessary manipulations might be easy even when exact manipulations are hard. For example, in the proof of Theorem~\ref{thm:five}, suppose that agent $a_{3\cdot N+1}$ has cardinal utility for item $x$ that is strictly greater than $(3^N).(3N+1)$. If they bid sincerely, their expected utility is $2$. If they bid strategically zero for item $z$, their expected utility is strictly greater than $2$. This \emph{necessary} increase can be decided in polynomial time but computing the \emph{exact} increase is intractable. However, necessary manipulations are also in general not always easy even if we ask merely for any increase in the expected utility of a given agent. 

\begin{mytheorem}\label{thm:six}
With $n>2$ agents, 0/1 utilities and the \myblike\ mechanism, problem {\sc NecessaryManipulation} is in $\coNP$-hard under Turing reductions.
\end{mytheorem}

\begin{proof}
Consider instance $\mathcal{P}_{G,r}$. Suppose all agents bid sincerely. Hence, ${u}_{c(3N+M-r+1)}$ $({u}^{c},o)=p_{c}(z_{N-r},o)+p_c(w,o)$. Suppose all agents bid sincerely except agent $c$ who bids strategically 0 for item $w$. Let their bidding vector be ${v}^{c}$. We have that ${u}_{c(3N+M-r+1)}({v}^{c},o)=p_{c}(z_{N-r},o)$. The instance of {\sc NecessaryManipulation} uses as input instance $\mathcal{P}_{G,r}$, agent $c$, vectors ${v}^{c}$ and ${u}^{c}$, and rational number $k=0$. We conclude that ${u}_{c(3N+M-r+1)}({v}^{c},o)-{u}_{c(3N+M-r+1)}({u}^{c},o)\geq 0$ iff $p_c(w,o)=0$. The result follows by Theorem~\ref{thm:three}. \hfill \myqed
\end{proof}

Another definition of the manipulation problem is whether a player can possibly increase their utility by insincere reporting, rather than computing the necessary or exact gain. Observe that in the proof of Theorem~\ref{thm:six}, we have that ${u}_{c(3N+M-r+1)}({u}^{c},o)-{u}_{c(3N+M-r+1)}({v}^{c},o)>0$ iff $p_c(w,o)>0$. We conclude that possible manipulations are also intractable in general by the proof of Theorem~\ref{thm:three}. Finally, by Theorem~\ref{thm:four}, we conclude that possible, necessary and exact manipulations are easy with just two agents and items arriving from a fixed ordering. By Theorem~\ref{thm:one} and the discussion after it, we conclude that necessary and exact manipulations are hard with two agents and items arriving from a distribution.

\section{Related Work and Conclusion}\label{sec:rel} 

We studied the worst-case computational complexity of possible, necessary and exact outcomes returned by the \mylike\ and \myblike\ mechanisms supposing agents act sincerely. With \mylike, there is no benefit for agents to act strategically. With \myblike, the agents might be strategic but we proved that computing a manipulation is computationally intractable in general. Some results are however tractable for the case of 2 agents. Our study of the online allocations returned by the \mylike\ and
\myblike\ mechanisms is in-line with many results in offline fair
division, voting theory and partial tournaments where possible, necessary and exact outcomes play crucial role;
see e.g. \cite{aziz2015pnw,aziz2015pn,bachrach2010,xia2011}. Our results provide a stepping stone towards better understanding strategic behavior. A number of works already considered such behavior for offline mechanisms; see e.g. \cite{aziz2015e,bouveret2014}. Another interesting future directions would be to estimate the outcomes of our mechanisms or to look at fixed-parameter tractable algorithms for these problems \cite{downey2013,jerrum1989,saban2015}.

\bibliographystyle{splncs}
\bibliography{outcomes}

\end{document}